\documentclass[onecolumn]{article}
\usepackage{amssymb,amsmath,amsthm}
\usepackage{mathrsfs}
\usepackage[pdftex,pdfpagelabels,plainpages=false]{hyperref}
\usepackage{url}
\usepackage{bbm}   
\usepackage{graphicx}
\usepackage{geometry}
\usepackage[final]{changes}

\usepackage[inline,marginclue,final,silent]{fixme}

\usepackage{paralist}

\usepackage[all]{xy}
\input xy
\xyoption{all}
\xyoption{arc}
\xyoption{line}

\newcommand{\beq}{\begin{displaymath}}
\newcommand{\eeq}{\end{displaymath}}
\newcommand{\beqn}{\begin{equation}}
\newcommand{\eeqn}{\end{equation}}
\newcommand{\beqa}{\begin{eqnarray*}}
\newcommand{\eeqa}{\end{eqnarray*}}
\newcommand{\beqna}{\begin{eqnarray}}
\newcommand{\eeqna}{\end{eqnarray}}

\newcommand{\re}[1]{~(\ref{#1})}

\newcommand{\N}{\mathbb{N}}
\newcommand{\Z}{\mathbb{Z}}
\newcommand{\R}{\mathbb{R}}
\newcommand{\C}{\mathbb{C}}

\newcommand{\ra}{\rightarrow}

\newcommand{\lra}{\longrightarrow}

\newcommand{\tr}{\mathrm{tr}}

\newtheorem{proposition}{Proposition}[section]
\newtheorem{theorem}[proposition]{Theorem}

\newtheorem{lemma}[proposition]{Lemma}

\theoremstyle{definition} 
\newtheorem{example}[proposition]{Example}
\newtheorem{remark}[proposition]{Remark}

\newtheorem{question}[proposition]{Question}
\title{\replaced{On the existence of quantum representations\\ for two dichotomic measurements}{Representing two binary repeatable\\ measurements on Hilbert space}}
\author{Tobias Fritz\\
Max Planck Institute for Mathematics\\
\texttt{fritz@mpim-bonn.mpg.de}}

\begin{document}

\maketitle

\begin{abstract}
\replaced{Under which conditions do outcome probabilities of measurements possess a quantum-mechanical model? This kind of problem is solved here for the case of two dichotomic von Neumann measurements which can be applied repeatedly to a quantum system with trivial dynamics. The solution uses methods from the theory of operator algebras and the theory of moment problems. The ensuing conditions reveal surprisingly simple relations between certain quantum-mechanical probabilities. It also shown that generally, none of these relations holds in general probabilistic models. This result might facilitate further experimental discrimination between quantum mechanics and other general probabilistic theories.}{When do measurements have a quantum-mechanical model? This article thi kind of inverse problem for iterated measurements in quantum mechanics and presents a solution of the simplest non-trivial case. Before that, preliminary observations on pre- and postselection reveal some unexpected properties of quantum-mechanical probabilities, which might potentially be tested in experiments. Then, the simplest non-trivial case of the inverse problem, corresponding to two binary repeatable measurements, is analyzed and solved using methods from operator algebras and the theory of moment problems. Also, the quantum region is studied in some truncations. After that, it is shown that the general probabilistic region spans the whole space of probabilities. This undermines the fact that quantum theories are a very special type of general probabilistic theories. The article ends with some remarks on the complexity of the general inverse problem for quantum measurements and an outline of some properties that potential experimental tests of the quantum constraints should have.}
\end{abstract}


\section{Introduction}

Consider the following situation: an experimenter works with some fixed physical system whose theoretical description is assumed to be unknown. In particular, it is not known whether the system obeys the laws of quantum mechanics or not. Suppose \replaced{also}{now} that the experimenter can \replaced{conduct}{do} two different types of measurement\added{---call them} $a$ and $b$ \replaced{---}{,} each of which is \added{\textit{dichotomic}, i.e.} has the possible outcomes $0$ and $1$. In this \replaced{paper}{article}, such a\deleted{s} system will be referred to as the ``black box figure~\ref{bb}''.

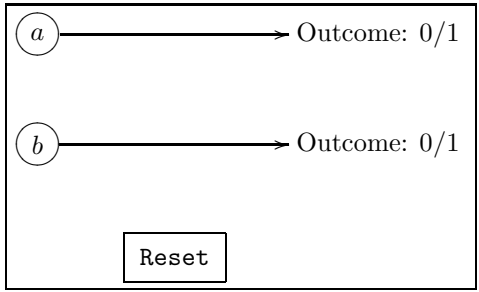
\begin{figure}
\caption{A black box with two \replaced{dichotomic}{binary} measurements and an initialization button.}
\label{bb}
\beq
\fbox{$\xymatrix{
*++[o][F]{a}\ar[rr] && \textrm{Outcome: }0/1\\
*++[o][F]{b}\ar[rr] && \textrm{Outcome: }0/1\\
& *++[F]{\texttt{Reset}}}$}
\eeq
\end{figure}

The experimenter can conduct several repeated measurements on the same system---like first $a$, then $b$, and then again $a$---and also he can conduct many of these repeated measurements on independent copies of the original system by hitting the ``$\mathtt{Reset}$'' button and starting over. \replaced{Thereby}{In this way}, he will obtain his results in terms of estimates for probabilities of the form
\beqn
\label{probex}
P_{a,b,a}(1,0,0)
\eeqn
which stands for the probability of obtaining the \added{sequence of} outcomes $1$, $0$, $0$, given that he first measures $a$, then $b$, and then again $a$.

Now suppose that the experimenter finds out that the measurements $a$ and $b$ are always repeatable, in the sense that measuring \replaced{one of them}{twice} consecutively yields \replaced{always the same result}{identical results} with certainty. In his table of experimentally determined probabilities, this is registered by statements like $P_{b,a,a,b}(0,1,0,0)=0$.

In a quantum-mechanical description of the system, the repeatable measurements $a$ and $b$ \replaced{are}{will} each \deleted{be} represented by projection operators on some Hilbert space $\mathcal{H}$ and the initial state of the system is given by some state on $\mathcal{H}$; it is irrelevant whether this state is assumed to be pure or mixed, since both cases can be reduced to each other: every pure state is trivially mixed, and a mixed state can be purified by entangling the system with an ancilla. In any case, the probabilities \added{like}\re{probex} can be calculated from this data by the usual rules of quantum mechanics.

\begin{question}
Which conditions do these probabilities $P_\cdot(\cdot)$ have to satisfy in order for a quantum-mechanical description of the system to exist?
\end{question}

Mathematically, this is a certain moment problem in noncommutative probability theory. \deleted{In its solution, I have tried to be as rigorous as possible.} Physically, the constraints turn out to be so unexpected that an intuitive explanation of their presence seems out of reach. 

\added{A variant of this problem has been studied by Khrennikov~\cite{Kh}, namely the case of two observables $a$ and $b$ with discrete non-degenerate spectrum. In such a situation, any post-measurement state is uniquely determined by the outcome of the directly preceding measurement. Hence in any such quantum-mechanical model, the outcome probabilities of an alternating measurement sequence $a,b,a,\ldots$ form a Markov chain, meaning that the result of any intermediate measurement of $a$ (respectively $b$) depends only on the result of the directly preceding measurement of $b$ (repectively $a$). Furthermore, by symmetry of the scalar product $|\langle\psi|\varphi\rangle|^2=|\langle\varphi|\psi\rangle|^2$, the corresponding matrix of transition probabilities is symmetric and doubly stochastic. In the case of two dichotomic observables, non-degenarcy of the spectrum is an extremely restrictive requirement; in fact, a dichotomic observable is necessarily degenerate as soon as the dimension of its domain is at least $3$. It should then not be a surprise that neither the Markovianness nor the symmetry and double stochasticity hold in general, making the results presented in this paper vastly more complex than Khrennikov's.}

\paragraph{Summary.} This \replaced{paper}{article} is structured as follows. Section~\ref{prelim} begins by generally studying a \replaced{dichotomic}{binary} quantum measurement under the conditions of pre- and postselection. It is found that both outcomes are equally likely, provided that the postselected state is orthogonal to the preselected state. Section~\ref{twobinrep} goes on by settling notation and terminology for the probabilities in the black box figure~\ref{bb} and describes the space of all conceivable outcome probability distributions for such a system. The main theorem describing the quantum region within this space is stated and proven in section~\ref{quantum}\replaced{. The largest part of this section}{, the largest part of which} is solely devoted to the theorem's technical proof; some relevant mathematical background material on moment problems can be found in the appendix~\ref{app}. Section~\ref{truncations} then studies projections of the space of all \added{conceivable outcome} probabilities and mentions some first results on the quantum region therein; \replaced{these finite-dimensional projections}{this} would mostly be relevant for potential experimental tests. Section~\ref{gp} \replaced{continues}{goes on} by proving that every point in the whole \replaced{space of all conceivable outcome probability distributions}{probability space} has a model in terms of a general probabilistic theory. As described in section~\ref{gens}, determining the quantum region for a higher number of measurements or a higher number of outcomes should be expected to be very hard. Section~\ref{experiments} mentions some properties that experiments comparing quantum-mechanical models to \replaced{different}{other} general probabilistic models should have. Finally, section~\ref{concl} briefly concludes the \replaced{paper}{article}.

\paragraph{Acknowledgements.} I want to thank Andrei Khrennikov for organizing a very inspiring conference ``Quantum Theory: Reconsideration of Foundations 5'' in V\"axj\"o. During discussions, I have received useful input from Cozmin Ududec, who encouraged me to think about iterated measurements in general probabilistic theories, as well as from Ingo Kamleitner, who suggested the quantum dot experiment described in section~\ref{prelim}. \added{I have also highly profited from conversations with Fabian Furrer and Wojciech Wasilewski. Finally, this work would not have been possible without the excellent research conditions within the IMPRS graduate program.}

\paragraph{Notation and terminology.} 

\replaced{Given a projection operator $p$, its}{Quantum-mechanically, a binary repeatable measurement corresponds to a projection operator $p$. Its} negation is written as $\overline{p}\equiv 1-p$. In order to have a compact index notation for $p$ and $\overline{p}$ at once, I \added{will} also write $p^1=p$ and $p^0=\overline{p}=1-p$, \added{which indicates that $p^1$ is the eigenspace projection corresonding to the measurement outcome $1$, while $p^0$ is the eigenspace projection corresponding to the measurement outcome $0$.}

The Pauli matrices
\beq
\sigma_x=\left(\begin{array}{cc}0&1\\1&0\end{array}\right),\quad\sigma_z=\left(\begin{array}{cc}1&0\\0&-1\end{array}\right),
\eeq
will be used in section~\ref{quantum} and in the appendix.

\added{Finally, $\{0,1\}^*\equiv\cup_{n\in\N}\{0,1\}^n$ is the set of all binary strings of arbitrary length.}

\section{Preliminary observations}
\label{prelim}

Before turning to the general case, this section \replaced{presents}{derives} some results about \added{outcome probabilities for} the \deleted{iterated} measurement \added{sequence} $a,b,a$ and reveals some unexpected constraints for quantum-mechanical models. One may think of the \added{two} measurements of $a$ in $a,b,a$ as being pre- and postselection, respectively, for the intermediate measurement of $b$.

So to ask a slightly different question first: how does a general quantum-mechanical \replaced{dichotomic}{binary} measurement $b$ behave under conditions of pre- and postselection? Suppose we conduct an experiment which 
\begin{itemize}
\item preselects with respect to a state $|\psi_i\rangle$, i.e. \added{initially,} it conducts a \deleted{initial} measurement of the projection \added{operator} $|\psi_i\rangle\langle\psi_i|$ and starts over in case of a negative result, and
\item postselects with respect to a state $|\psi_f\rangle$ i.e., it \added{finally} conducts a \deleted{final} measurement of the projection operator $|\psi_f\rangle\langle\psi_f|$ and starts \added{all} over from the beginning in case of a negative result.
\end{itemize}
In between the pre- and the postselection, the experimenter measures the \replaced{dichotomic}{binary} observable $b$. For simplicity, the absence of any additional dynamics is assumed.

This kind of situation \replaced{can only occur}{only makes sense} when the final postselection does not always produce a negative outcome, so that the conditional probabilities with respect to pre- and postselection have definite values. This is the case if and only if
\beq
\langle\psi_i|b|\psi_f\rangle\neq 0\quad\added{\textrm{or}}\quad\langle\psi_i|(1-b)|\psi_f\rangle\neq 0,
\eeq
which will be assumed to hold from now on; under the assumption of the following proposition, these two conditions are equivalent.

\begin{proposition}
\label{preorthpost}
\added{In such a situation}, \replaced{the condition}{In the case that} $\langle\psi_i|\psi_f\rangle=0$ \added{implies that} the two outcomes of $b$ have equal probability, independently of any details of the particular quantum-mechanical model:
\beqa
&&P\left(b=0\:\Big|\:\textrm{pre}=|\psi_i\rangle,\textrm{post}=|\psi_f\rangle\right)\\
&=&P\left(b=1\:\Big|\:\textrm{pre}=|\psi_i\rangle,\textrm{post}=|\psi_f\rangle\right)\:\:\:=\:\:\:\frac{1}{2}
\eeqa
\end{proposition}

Note that such a \added{pre- and postselected} \replaced{dichotomic}{binary} quantum measurement would \replaced{therefore be}{yield} a perfectly unbiased random number generator. 

\begin{proof}
The proof of proposition~\ref{preorthpost} is by straightforward calculation. Upon preselection, the system is in the state $|\psi_i\rangle$. The probability of measuring $b=0$ and successful postselection is given by
\beqa
||\:|\psi_f\rangle\langle\psi_f|(1-b)|\psi_i\rangle||^2&=&\langle\psi_i|(1-b)|\psi_f\rangle\langle\psi_f|(1-b)|\psi_i\rangle\\
&=&-\langle\psi_i|b|\psi_f\rangle\langle\psi_f|(1-b)|\psi_i\rangle\\
&=&\langle\psi_i|b|\psi_f\rangle\langle\psi_f|b|\psi_i\rangle\\
&=&||\:|\psi_f\rangle\langle\psi_f|b|\psi_i\rangle||^2.
\eeqa
\replaced{This}{which} equals the probability of measuring $b=1$ and successful postselection, \added{so that both conditional probabilities equal $1/2$.}
\end{proof}

\replaced{As}{For} a concrete example, consider a quantum particle which can be located in either of three boxes $|1\rangle$, $|2\rangle$, and $|3\rangle$, so that the state space is given by
\beq
\mathcal{H}=\C^3=\mathrm{span}\left\{|1\rangle,|2\rangle,|3\rangle\right\}
\eeq
Now let $\zeta$ be a third root of unity, such that $1+\zeta+\zeta^2=0$, and use initial and final states as follows:
\beqa
&\textrm{preselection: }|\psi_i\rangle=\frac{|1\rangle+|2\rangle+|3\rangle}{\sqrt{3}}\\\\
&\xymatrix{*++[F]{\textrm{box }|1\rangle} & *++[F]{\textrm{box }|2\rangle} & *++[F]{\textrm{box }|3\rangle}}\\\\
&\textrm{postselection: }|\psi_f\rangle=\frac{|1\rangle+\zeta|2\rangle+\zeta^2|3\rangle}{\sqrt{3}}
\eeqa
\replaced{Take the intermediate dichotomic measurement to be given by opening one of the boxes and checking whether the particle is there. This}{Then upon opening any box as the intermediate measurement and detecting presence of the particle} will locate the particle in that box with a (conditional) probability of \emph{exactly} $1/2$; see~\cite{Aha} for the original version of this three-boxes thought experiment, with even more counterintuitive consequences. Possibly such an experiment might be realized \added{in a way similar to the optical realization of the original Aharanov-Vaidman thought experiment~\cite{RLS} or by} using quantum dots as boxes. And possibly a high-precision version of such an experiment---looking for deviations from the quantum prediction of exactly $1/2$---might be an interesting further experimental test of quantum mechanics. In order to guarantee the crucial assumption of exact orthogonality of initial and final states, one could implement both pre- and postselection via the same von Neumann measurement and select for a final outcome differing from the initial outcome.

A similar calculation as in the proof of proposition~\ref{preorthpost} also shows that the following more general statement is true:

\begin{proposition}
\label{genpreorthpost}
\begin{enumerate}
\item Given \replaced{any discrete}{an} observable $a$ together with two different eigenvalues $\lambda_0\neq\lambda_1$ and a projection observable $b$, the outcome probabilities for $b$ under $(a=\lambda_0)$-preselection and $(a=\lambda_1)$-postselection are equal:
\beq
P_b\left(0\:\Big|\:a_{\mathrm{pre}}=\lambda_0,a_{\mathrm{post}}=\lambda_1\right)=P_b\left(1\:\Big|\:a_{\mathrm{pre}}=\lambda_0,a_{\mathrm{post}}=\lambda_1\right)=\frac{1}{2}
\eeq
\item The same holds true upon additional preselection before the first measurement of $a$, and also upon additional postselection after the second measurement of $a$.
\end{enumerate}
\end{proposition}

\added{So} what does \added{all} this imply for quantum-mechanical models of the black box figure~\ref{bb}? Given that one measures the sequence $a,b,a$ such that the two measurements of $a$ yield $0$ and $1$ respectively, then the two outcomes for $b$ have equal probability:
\beqn
\label{aba}
\fbox{$P_{a,b,a}(0,0,1)=P_{a,b,a}(0,1,1)$}
\eeqn
Similar relations can be obtained \replaced{from}{form} this equation by permuting $a\leftrightarrow b$ and $0\leftrightarrow 1$. In words: given that the second measurement of $a$ has a result different from the first, then the intermediate \replaced{dichotomic}{binary} measurement of $b$ has conditional probability $1/2$ for each outcome, no matter what the physical details of the quantum system are and what the initial state is. This is trivially true in the case that $a$ and $b$ commute: then, both probabilities in\re{aba} vanish.

\section{Probabilities for two \replaced{dichotomic}{binary} repeatable measurements}
\label{twobinrep}

In the situation of figure~\ref{bb}, the repeatability assumption for both $a$ and $b$ has the consequence that it is sufficient to consider alternating measurements of $a$ and $b$ only. Therefore, all non-trivial outcome probabilities are encoded in the following two stochastic processes:
\beq
P_{a,b,a,\ldots}(\ldots)
\eeq
and
\beq
P_{b,a,b,\ldots}(\ldots).
\eeq
Both of these expressions are functions taking a finite binary string in $\{0,1\}^*$ as their argument, and returning the probability of that outcome for the specified sequence of alternating measurements. \added{In the rest of this paper, the probabilities of the form $P_{a,b,a,\ldots}$ will be denoted by $P_a$ for the sake of brevity, while similarly $P_b$ stands for the probabilities determining the second stochastic process $P_{b,a,b,\ldots}$.}

Since total probability is conserved, it is clear that for every finite binary string $r\in\{0,1\}^*$,
\beqn
\label{conserved}
\begin{split}
P_{a\deleted{,b,a,\ldots}}(r)&=P_{a\deleted{,b,a,\ldots}}(r,0)+P_{a\deleted{,b,a,\ldots}}(r,1)\\
P_{b\deleted{,a,b,\ldots}}(r)&=P_{b\deleted{,a,b,\ldots}}(r,0)+P_{b\deleted{,a,b,\ldots}}(r,1)
\end{split}
\eeqn
 A probability assignment for the $P_a$'s and $P_b$'s is called \emph{admissible} whenever the probability conservation laws\re{conserved} hold.

\section{Classification of probabilities in quantum theories}
\label{quantum}

Now let us assume that the black box figure~\ref{bb} does have a quantum-mechanical description and determine \added{all} the constraints that then \added{have to} hold for the probabilities $P_a$ and $P_b$.

The final results \replaced{will be}{are} presented right now at the beginning. \deleted{and} The rest of the section is then devoted to showing how this theorem can be derived from the mathematical results presented in the appendix.

Given a binary \replaced{string}{sequence} $r\in\{0,1\}^n$, denote the number of switches in $r$ by $s(r)$, i.e. the number of times that a $1$ follows a $0$ or a $0$ follows a $1$. \deleted{Similarly, $s(r,r')$ is the number of switches in the concatenated string $r,r'$.} The single letter $r$ and the sequence $r_1,\ldots,r_n$ are interchangeable notation \added{for the same binary string}.

The overline notation $\overline{r}$ stands for the inverted \replaced{string}{sequene}, i.e. $0\leftrightarrow 1$ in $r$. The letter $\mathcal{C}$ denotes the convex subset of $\R^4$ that is defined and characterized in the appendix.

\begin{theorem}
\label{mainthm}
A quantum-mechanical description of the black box figure~\ref{bb} exists if and only if the outcome probabilities satisfy the following constraints:
\begin{itemize}
\item \added{For every $r\in\{0,1\}^{n+1}$ and $i\in\{a,b\}$, the probabilities} 
\beq
P_{i}(r_1,\ldots,r_{n+1})
\eeq
only depend on $i$, $\added{s=}s(r)$ and $r_1$; denote this value by $\deleted{P_{i}(r_1,\ldots,r_{n+1})=}F_{i,r_1}(n,s)$. \deleted{Note of change: these were previously four separate items on the list of constraints.}
\item \replaced{For every $r\in\{0,1\}^n$,}{The equation}
\beq
P_{a}(r)+P_{a}(\overline{r})=P_{b}(r)+P_{b}(\overline{r}).
\eeq
\deleted{holds for all $r$.}
\item Using the notation
\beqa
F_{a,+}(n,s)&=&F_{a,1}(n,s)+F_{a,0}(n,s)\\
C_1(n,s)&=&\frac{1}{2}\left(F_{a,1}(n,s)-F_{a,0}(n,s)+F_{b,1}(n,s)-F_{b,0}(n,s)\right)\\
C_2(n,s)&=&\frac{1}{2}\left(F_{a,1}(n,s)-F_{a,0}(n,s)-F_{b,1}(n,s)+F_{b,0}(n,s)\right),
\eeqa
the inequality\footnote{Note that all sums are automatically absolutely convergent since $F_{\cdot,\cdot}(\cdot,\cdot)\in[0,1]$ and $\sum_{k=0}^\infty\left|\binom{1/2}{k}\right|=1<\infty$.}
\beq
\begin{split}
&\left(\sum_{k=0}^\infty(-1)^k\binom{1/2}{k}C_1(n+k-1,s+k)\right)^2\\
+&\left(\sum_{k=0}^\infty(-1)^k\binom{1/2}{k}C_2(n+k-1,s\added{-1})\right)^2\leq F_{a,+}(n,s)^2
\end{split}
\eeq
holds for every $n\in\N$ and $s\in\{1,\ldots,n-1\}$.
\item Using the coefficients
\beq
c_{n,k}=(-1)^k\binom{-1/2}{k}-(-1)^{k-n}\binom{-1/2}{k-n}
\eeq
and the quantities
\beq
\begin{split}
V_{x,\pm}(n)=&\sum_{k=0}^\infty c_{n,k}C_1(k,k)\\
&\pm\sqrt{F_{a,+}(n,n)^2-\left(\sum_{k=0}^\infty(-1)^k\binom{1/2}{k}C_2(n+k-1,n-1)\right)^2}\\
V_{z,\pm}(n)=&\sum_{k=0}^\infty c_{n,k}C_2(k,0)\\
&\pm\sqrt{F_{a,+}(n,0)^2-\left(\sum_{k=0}^\infty(-1)^k\binom{1/2}{k}C_1(n+k-1,k)\right)^2},
\end{split}
\eeq
the point in $\R^4$ given by
\beqn
\left(\sup_{n}V_{x,-}(n),\:\sup_nV_{z,-}(n),\:\inf_{n}V_{x,+}(n),\:\inf_nV_{z,+}(n)\right)
\eeqn
has to lie in the convex region $\mathcal{C}\subseteq\R^4$ characterized in proposition~\ref{C}.\footnote{In particular, the expressions under the square roots have to be non-negative and the suprema and infima have to be finite.}
\end{itemize}
\end{theorem}

To begin the proof of this theorem, let $\mathcal{A}_2=C^*(a,b)$ be the $C^*$-algebra freely generated by two projections $a$ and $b$. Then for every quantum-mechanical model of the system, we obtain a unique $C^*$-algebra homomorphism
\beq
\mathcal{A}_2\lra\mathcal{B}(\mathcal{H})
\eeq
which maps the universal projections to concrete projections on $\mathcal{H}$. \deleted{Then} Upon pulling back the black \replaced{box's}{boxes'} initial state $|\psi\rangle$ to a $C^*$-algebraic state on $\mathcal{A}_2$, we can calculate all outcome probabilities via algebraic quantum mechanics on $\mathcal{A}_2$. Conversely, any $C^*$-algebraic state on $\mathcal{A}_2$ defines a quantum-mechanical model of the two \replaced{dichotomic}{binary} observables system by virtue of the GNS construction. Therefore, we will do all further considerations on $\mathcal{A}_2$. In this sense, the states on $\mathcal{A}_2$ are the universal instances of quantum black boxes figure~\ref{bb}.

$\mathcal{A}_2$ is known~\cite{RS} to be of the form
\beq
\mathcal{A}_2\cong\left\{f:[0,1]\stackrel{\added{\textrm{cont.}}}{\lra} M_2(\C)\:\big|\:f(0),\,f(1)\textrm{ are diagonal}\right\}
\eeq
where the universal pair of projections is given by
\beqa
a(t)&=&\left(\begin{array}{cc}1&0\\0&0\end{array}\right)=\frac{\mathbbm{1}_2+\sigma_z}{2}\\
b(t)&=&\left(\begin{array}{cc}t&\sqrt{t(1-t)}\\\sqrt{t(1-t)}&1-t\end{array}\right)=\frac{1}{2}\mathbbm{1}_2+\sqrt{t(1-t)}\,\sigma_x+\left(t-\frac{1}{2}\right)\sigma_z
\eeqa
By the Hahn-Banach extension theorem, the set of states on $\mathcal{A}_2$ can be identified with the set of functionals obtained by restricting the states on the full \deleted{matrix} algebra \added{of matrix-valued continuous functions} $\mathscr{C}\left([0,1],M_2(\C)\right)$ to the subalgebra $\mathcal{A}_2$. Hence for the purposes of the proof of theorem~\ref{mainthm}, there is no need to distinguish between $\mathcal{A}_2$ and $\mathscr{C}\left([0,1],M_2(\C)\right)$.

Now consider a sequence of $n+1$ sequential measurements having the form $a,b,a,\ldots$. The set of outcomes for all measurements taken together is given by the set $\{0,1\}^{n+1}$ of \replaced{dichotomic}{binary} \replaced{strings}{sequences} $r=\left(r_i\right)_{i=\replaced{1}{0}}^{n\added{+1}}$. Every such outcome $r$ has an associated Kraus operator which is given by
\beqn
\label{kraus}
H_r=a^{r_{\replaced{1}{0}}}b^{r_{\replaced{2}{1}}}a^{r_{\replaced{3}{2}}}\ldots
\eeqn
where the superscripts indicate whether one has to insert the projection $a$ or $b$ itself or its orthogonal complement $\overline{a}=1-a$ or $\overline{b}=1-b$, respectively. Then the probability of obtaining \replaced{the string $r$ as an}{that} outcome is given by the expression
\beqn
\label{outcomeprob}
\begin{split}
P_a\left(r_{\replaced{1}{0}},\deleted{r_1,}\ldots,r_{n\added{+1}}\right)&=\rho\left(H_rH_r^\dagger\right)\\
&=\rho\left(a^{r_{\replaced{1}{0}}}b^{r_{\replaced{2}{1}}}a^{r_{\replaced{3}{2}}}\ldots a^{r_{\replaced{3}{2}}}b^{r_{\replaced{2}{1}}}a^{r_{\replaced{1}{0}}}\right)
\end{split}
\eeqn
Now follows the main observation which facilitates all further calculations.

\begin{lemma}
We have the following reduction formulas in $\mathcal{A}_2$:
\beqa
\begin{array}{lll}
aba=ta,&\quad& bab=tb\\
a\overline{b}a=(1-t)a,&\quad& b\overline{a}b=(1-t)b\\
\overline{a}b\overline{a}=(1-t)\overline{a},&\quad&\overline{b}a\overline{b}=(1-t)\overline{b}\\
\overline{a}\overline{b}\overline{a}=t\overline{a},&\quad&\overline{b}\overline{a}\overline{b}=t\overline{b}
\end{array}
\eeqa
\end{lemma}

\begin{proof}
Direct calculation.
\end{proof}

As a consequence, one finds that the measurement outcome probabilities\re{outcomeprob} have the form
\beq
P_a\left(r_{\replaced{1}{0}},\deleted{r_1,}\ldots,r_{n\added{+1}}\right)=\rho\left(t^{n-s}(1-t)^s a^{r_1}\right)
\eeq
where $s$ is the number of switches in the \replaced{dichotomic}{binary} \replaced{string}{sequence} $r_0,\ldots,r_n$;
the same clearly applies to the $P_b$'s that determine the outcome
probabilities for the measurement sequence $b,a,b,\ldots$. Hence, one necessary
condition on the probabilities is the following:

\begin{proposition}
The probabilities $P_a\left(r_{\replaced{1}{0}},\deleted{r_1,}\ldots,r_{n\added{+1}}\right)$ only depend on the number of switches contained in the \replaced{dichotomic}{binary} sequence $r_{\replaced{1}{0}},\ldots,r_{n\added{+1}}$. The same holds for the $P_b\left(r_{\replaced{1}{0}},\deleted{r_1,}\ldots,r_{n\added{+1}}\right)$.
\end{proposition}

A particular instance of this is equation\re{aba}.

\begin{remark}
Moreover, this observation is actually a \emph{consequence} of the conditional statement of proposition~\ref{genpreorthpost}(b). Due to that result, it is clear that the equations
\beq
P_a(r_{\replaced{1}{0}},\ldots,r_k,0,0,1,r_{k+3},\ldots,r_{n\added{+1}})=P_a(r_{\replaced{1}{0}},\ldots,r_k,0,1,1,r_{k+3},\ldots,r_{n\added{+1}})
\eeq
\beq
P_a(r_{\replaced{1}{0}},\ldots,r_k,1,0,0,r_{k+3},\ldots,r_{n\added{+1}})=P_a(r_{\replaced{1}{0}},\ldots,r_k,1,1,0,r_{k+3},\ldots,r_{n\added{+1}})
\eeq
hold. In words: the outcome probability does not change if the position of a switch in the binary string is moved by one. On the other hand, \deleted{it is clear that} any two binary sequences with the same number of switches can be transformed into each other by subsequently moving the position of each switch by one.
\end{remark}

Since the dependence on the sequence $r$ is only via its length $n+1$, the number of switches $s$, and the initial outcome $r_1$, mention of $r$ will be omitted from now on. Instead, \added{the dependence on $r$ will be retained by considering} all expressions \deleted{should be considered} as functions of $n$, $r_1$ and $s$, with $s\in\{0,\ldots,n\}$. The two possible values of the initial outcome $r_1$ as well as the initial type of measurement are indicated by subscripts:
\beqa
P_a(0,r_{\replaced{2}{1}},\ldots,r_{n\added{+1}})=F_{a,0}(n,s)\\
P_a(1,r_{\replaced{2}{1}},\ldots,r_{n\added{+1}})=F_{a,1}(n,s)\\
P_b(0,r_{\replaced{2}{1}},\ldots,r_{n\added{+1}})=F_{b,0}(n,s)\\
P_b(1,r_{\replaced{2}{1}},\ldots,r_{n\added{+1}})=F_{b,1}(n,s)
\eeqa
By the present results, the four functions $F_{\cdot,\cdot}$ can be written as
\beqa
F_{a,1}(n,s)=\rho\left(t^{n-s}(1-t)^sa\right)\\
F_{a,0}(n,s)=\rho\left(t^{n-s}(1-t)^s\overline{a}\right)\\
F_{b,1}(n,s)=\rho\left(t^{n-s}(1-t)^sb\right)\\
F_{b,0}(n,s)=\rho\left(t^{n-s}(1-t)^s\overline{b}\right)
\eeqa
\replaced{But actually}{Though} instead of using these sequences of probabilities, the patterns are easier to spot \added{when} using \deleted{instead} the new variables
\beqa
F_{a,+}(n,s)\equiv F_{a,1}(n,s)+F_{a,0}(n,s),&\quad& F_{a,-}(n,s)\equiv F_{a,1}(n,s)-F_{a,0}(n,s)\\
F_{b,+}(n,s)\equiv F_{b,1}(n,s)+F_{b,0}(n,s),&\quad& F_{b,-}(n,s)\equiv F_{b,1}(n,s)-F_{b,0}(n,s)
\eeqa
In these terms, we can write the four equations as
\beq
\begin{split}
F_{a,+}(n,s)&=\rho\left(t^{n-s}(1-t)^s\right)\\
F_{b,+}(n,s)&=\rho\left(t^{n-s}(1-t)^s\right)\\
F_{a,-}(n,s)&=\rho\left(t^{n-s}(1-t)^s\sigma_z\right)\\
F_{b,-}(n,s)&=\rho\left(t^{n-s}(1-t)^s\left[2\sqrt{t(1-t)}\,\sigma_x+\left(2t-1\right)\sigma_z\right]\right)
\end{split}
\eeq
Therefore, it is clear that another necessary constraint is that 
\beq
F_{a,+}(n,s)=F_{b,+}(n,s)\quad\forall n,s
\eeq
In terms of the probabilities, this translates into
\beq
\fbox{$P_a(r)+P_a(\overline{r})=P_b(r)+P_b(\overline{r})$}
\eeq
The first non-trivial instance of this occurs for the case $n=1$, where we have the equations
\beqa
P_a(0,0)+P_a(1,1)&=&P_b(0,0)+P_b(1,1)\\
P_a(0,1)+P_a(1,0)&=&P_b(0,1)+P_b(1,0)
\eeqa
which also have been noted in~\cite[p. 257/8]{AS}.

Finally, let us try to extract the conditions that need to be satisfied by the $F_{a,-}$ and $F_{b,-}$. Considering the form of the equations, it seems convenient to introduce the quantities
\beqa
C_1(n,s)&\equiv& \frac{1}{2}\left(F_{a,-}(n,s)+F_{b,-}(n,s)\right)\\
&=&\frac{1}{2}\left(F_{a,1}(n,s)-F_{a,0}(n,s)+F_{b,1}(n,s)-F_{b,0}(n,s)\right)\\
C_2(n,s)&\equiv&\frac{1}{2}\left(F_{a,-}(n,s)-F_{b,-}(n,s)\right)\\
&=&\frac{1}{2}\left(F_{a,1}(n,s)-F_{a,0}(n,s)-F_{b,1}(n,s)+F_{b,0}(n,s)\right)
\eeqa
which are \replaced{somewhat}{highly} reminiscient of the CHSH correlations. In these terms,
\beqa
C_1(n,s)&=&\rho\left(t^{n-s}(1-t)^s\underbrace{\left[\sqrt{t(1-t)}\,\sigma_x+t\sigma_z\right]}_{\vec{v}_1(t)\cdot\vec{\sigma}}\right)\\
C_2(n,s)&=&\rho\left(t^{n-s}(1-t)^s\underbrace{\left[-\sqrt{t(1-t)}\,\sigma_x+(1-t)\sigma_z\right]}_{\vec{v}_2(t)\cdot\vec{\sigma}}\right)\\
\eeqa
The reason that this is nicer is because now, the two vectors $\vec{v}_1(t)$, $\vec{v}_2(t)$, are orthogonal for each $t$. Finally, $\vec{v}_1(t)$ and $\vec{v}_2(t)$ can be normalized to get
\beqa
C_1(n,s)&=&\rho\left(t^{n-s+1/2}(1-t)^s\,\vec{n}_1(t)\cdot\vec{\sigma}\right)\\
C_2(n,s)&=&\rho\left(t^{n-s}(1-t)^{s+1/2}\,\vec{n}_2(t)\cdot\vec{\sigma}\right)
\eeqa
with vectors $\vec{n}_1(t)$, $\vec{n}_2(t)$, that are normalized and orthogonal for each $t$. Using an appropriate automorphism of \added{$\mathscr{C}([0,1],M_2(\C))$}\deleted{$\mathcal{A}_2$} given by conjugation with a $t$-dependent unitary $U(t)\in SU(2)$, the vectors $\vec{n}_i(t)$ can be rotated in such a way that they coincide with the standard basis vectors $\vec{e}_x$ and $\vec{e}_z$, constant as functions of $t$.

Then, theorem~\ref{mainthm} is a consequence of theorem~\ref{ncmp2} as applied to
\beqa
M'_1(n,s)&=&F_{a,+}(n,s)\\
M'_x(n,s)&=&C_1(n,s)\\
M'_z(n,s)&=&C_2(n,s).
\eeqa

\section{Determining the quantum region in truncations}
\label{truncations}

In actual experiments, only a finite number of the probabilities can be measured. Also, these can realistically only be known up to finite precision due to finite statistics. An even more problematic issue is that perfect von Neumann measurements are impossible to realize and can only be approximated. Here, we ignore the latter two problems and focus on the issue that only a finite number of probabilities are known.

\begin{question}
\label{qtrunc}
Given numerical values for a finite subset of the probabilities $P_\cdot(\cdot)$, how can one decide whether a quantum-mechanical \replaced{representation of}{model reproducing} these probabilities exists?
\end{question}

Clearly, such a \replaced{representation}{model} exists if and only if these probabilities can be extended to a specification of \emph{all} outcome probabilities $P_a$ and $P_b$ satisfying the conditions given in theorem~\ref{mainthm}. However, this observation doesn't seem very useful---how might one decide whether such an extension exists? The problem is that the projection of a convex set (the quantum region) from an \replaced{infinite}{higher}-dimensional vector space down to a \replaced{finite}{lower}-dimensional one can be notoriously hard to \replaced{compute}{describe}.

\replaced{Question~\ref{qtrunc}}{This problem} is a close relative of the truncated Hausdorff moment problem (see e.g.~\cite[ch. III]{Wi}). In a finite truncation of the Hausdorff moment problem, the allowed region coincides with the convex hull of the moments of the Dirac measures, which are exactly the extreme points in the space of measures. Therefore, the allowed region is the convex hull of an algebraic curve embedded in Euclidean space.

In the present case, it is possible to follow an analogous strategy of first determining the extreme points in the set of states---that is, the pure states on the algebra---\added{and} then calculat\added{ing} the corresponding points in the truncation, and finally tak\added{ing} the convex hull of this set of points. To begin this program, note that the pure states on the algebra are exactly those of the form
\beq
\mathscr{C}\left([0,1],M_2(\C)\right)\lra\C,\quad f\mapsto\langle\psi|f(t_0)|\psi\rangle
\eeq
where $t_0\in[0,1]$ is fixed, and $|\psi\rangle$ stands for some fixed unit vector in $\C^2$; this corresponds to integration with respect to a projection-valued Dirac measure on $[0,1]$. \deleted{Furthermore, since all the algebra elements of interest lie in the subalgebra given by the $M_2(\R)$-valued functions -- as opposed to $M_2(\C)$-valued -- the phases of the components of $|\psi\rangle$ are arbitrary, and can therefore be chosen such that $|\psi\rangle$ has real-valued components. This argument shows that} \added{Since global phases are irrelevant,} $|\psi\rangle$ can be assumed to be given by
\beq
|\psi\rangle = \left(\begin{array}{c}\cos\theta\\e^{i\lambda}\sin\theta\end{array}\right).
\eeq
In conclusion, the pure states are parametrized by the numbers $t_0\in[0,1]$\added{, $\lambda\in[0,2\pi]$} and $\theta\in[0,2\pi]$. In any given truncation, this determines an algebraic \replaced{variety}{surface}, whose convex hull coincides with the quantum region in that truncation. This reduces the problem\added{~\ref{qtrunc}} to the calculation of the convex hull of an algebraic \replaced{variety}{surface} embedded in Euclidean space.

The following theorem is concerned with the infinite-dimensional truncation to all $P_a$, which means that one simply disregards all probabilities $P_b$ while keeping the $P_a$.

\begin{theorem}
\label{atrunc}
A quantum-mechanical \replaced{representation}{model} in the $P_a$ truncation exists for an admissible probability assignment if and only if $P_a(r)$ only depends on $s(r)$.
\end{theorem}

\begin{proof}
It follows from the main theorem\re{mainthm} that this condition is necessary. To see that it is sufficient, recall the equations
\beqa
F_{a,+}(n,s)&=\rho\left(t^{n-s}(1-t)^s\right)\\
F_{a,-}(n,s)&=\rho\left(t^{n-s}(1-t)^s\sigma_z\right),
\eeqa
\added{which have been used in the proof of theorem~\ref{mainthm}.} Then upon choosing $M_1(n,s)=F_{a,+}(n,s)$, $M_x(n,s)=0$ and $M_z(n,s)=F_{a,-}(n,s)$, theorem~\ref{ncmp1} applies and shows that such a state $\rho$ can be found as long as the condition
\beq
|F_{a,-}(n,s)|\leq F_{a,+}(n,s)
\eeq
holds. In terms of the probabilities, this requirement means
\beq
\left|F_{a,1}(n,s)-F_{a,0}(n,s)\right|\leq F_{a,1}(n,s)+F_{a,0}(n,s),
\eeq
which always holds trivially since all probabilities are non-negative. This ends the proof.
\end{proof}

This ends the current treatment of truncations. It is hoped that the future study of truncations will be relevant for experiments.

\section{A general probabilistic model always exists}
\label{gp}

In order to understand as to how far the conditions found are characteristic of quantum mechanics, one should try to determine the analogous requirements for the probabilities in the case of alternative theories different from quantum mechanics and in the case of more general theories having quantum mechanics as a special case. This section deals with the case of general probabilistic theories. 

What follows is a brief exposition of the framework of general probabilistic theories and of the possible models for a black box system figure~\ref{bb}. Afterwards, it will be shown that every assignment of outcome probabilities for the black box system does have a general probabilistic model. Together with the results of the previous \added{two} sections, this shows that---for systems with two \replaced{dichotomic}{binary} measurements---quantum-mechanical models are a very special class of general probabilistic theories.

For the present purposes, a general probabilistic theory is defined by specifying a real vector space $V$, a non-vanishing linear functional $\tr:V\ra\R$, and a convex set of normalized states $\Omega\subseteq V$ such that
\beqn
\label{normalization}
\tr(\rho)=1\quad\forall\rho\in\Omega
\eeqn
The cone $\Omega_0\equiv\R_{\geq 0}\Omega$ is the set of all unnormalized states. By construction,
\beq
\Omega=\Omega_0\cap\tr^{-1}(1).
\eeq
Since all that matters for the physics is really $\Omega_0$ and $\tr$ on $\Omega_0$, one can assume without loss of generality that $\Omega_0$ spans $V$,
\beqn
\label{spansV}
V=\Omega_0-\Omega_0.
\eeqn

Now, an \emph{operation} is a linear map $T:V\ra V$ which maps unnormalized states to unnormalized states,
\beq
T(\Omega_0)\subseteq\Omega_0,
\eeq
and does not increase the trace,
\beq
\tr(T(\rho))\leq 1\quad\forall\rho\in\Omega.
\eeq
For $\rho\in\Omega$, the number $\tr(T(\rho))$ is interpreted as the probability that the operation takes place, given $T$ as one of several alternative operations characteristic of the experiment. In case that $T$ happens, the post-measurement state is given by
\beq
\rho'\equiv\frac{T(\rho)}{\tr(T(\rho))},
\eeq
where the denominator is just the normalization factor.

\begin{example}
As an example of this machinery, one may take density matrices as normalized states and completely positive \added{trace-nonincreasing} maps as operations. This is quantum theory; the usual form of a quantum operation in terms of Kraus operators can be recovered by virtue of the Stinespring factorization theorem.
\end{example}

A repeatable \replaced{dichotomic}{binary} measurement is then implemented by two operations $a,\overline{a}:V\ra V$ which are idempotent,
\beq
a^2=a,\quad\overline{a}^2=\overline{a},
\eeq
and complementary in the sense that the operation $a+\overline{a}$ preserves the trace. Physically, the operation $a$ takes place whenever the \added{dichotomic} measurement has \replaced{the outcome $1$}{a positive outcome}, whereas $\overline{a}$ happens in the case that the \added{dichotomic} measurement has \replaced{the outcome $0$}{a negative outcome}.

\begin{proposition}
Under these assumptions, $a\overline{a}=\overline{a}a=0$.
\end{proposition}

\begin{proof}
Clearly, $a\overline{a}$ is an operation, and therefore it maps $\Omega_0$ to $\Omega_0$. On the other hand,
\beq
\tr(\overline{a}a(\rho))=\tr(a(\rho))-\tr(aa(\rho))=0,
\eeq
which proves $\overline{a}a=0$ by\re{normalization}. The other equation works in exactly the same way.
\end{proof}

The interpretation of this result is that, when $a$ has been measured with \deleted{a positive} outcome\added{ $1$}, then the opposite result $\overline{a}$ will certainly not occur in \replaced{an immediately sequential}{the second} measurement, and vice versa. \replaced{In this sense}{Therefore}, the measurement of $a$ vs. $\overline{a}$ is repeatable. 

In the previous sections, the quantum region was found to be a very small subset of the space of all \added{admissible} \replaced{probability assignments}{probabilities}. The following theorem shows that this is not the case for general probabilistic theories.

\begin{theorem}
\label{gpthm}
Given any admissible probability assignment for the $P_a$'s and $P_b$'s, there exists a general probabilistic model that reproduces these probabilities.
\end{theorem}

\begin{proof}
The idea of the proof is analogous to the characterization of the quantum region done in section~\ref{quantum}: to try and construct a universal theory for the black box system, which covers all of the allowed region in probability space at once. In order to achieve category-theoretic universality (an initial object in the appropriate category), one \replaced{needs}{nneds} to consider the unital $\R$-algebra freely generated by formal variables $v_a$, $v_{\overline{a}}$, $v_b$, $v_{\overline{b}}$, subject to the relations imposed by the above requirements. Hence the definition is this,
\beqa
\mathcal{A}_{gp}=\Big\langle v_a,v_{\overline{a}},v_b,v_{\overline{b}}\:\:|&&v_av_{\overline{a}}=v_{\overline{a}}v_a=v_bv_{\overline{b}}=v_{\overline{b}}v_b=0,\\
&&v_a^2=v_a,\:v_{\overline{a}}^2=v_{\overline{a}},\:v_b^2=v_b,\:v_{\overline{b}}^2=v_{\overline{b}}\:\Big\rangle_{\R\mathtt{Alg}}
\eeqa
where the notation indicates that this is to be understood as a definition in terms of generators and relations in the category of unital associative algebras over the field $\R$. The index $gp$ stands for ``general probabilistic''. \added{This definition guarantees that any finite product of generators can be reduced to one of the form
\beq
v_{a^{r_1}}v_{b^{r_2}}v_{a^{r_3}}\ldots\qquad\textrm{or}\qquad v_{a^{r_1}}v_{b^{r_2}}v_{a^{r_3}}\ldots\:.
\eeq
These expressions, together with the unit $\mathbbm{1}$, form a linear basis of $\mathcal{A}_{gp}$.}

Now an unnormalized state on $\mathcal{A}_{gp}$ is defined to be a linear functional
\beq
\rho:\mathcal{A}_{gp}\lra\R
\eeq
which is required to be non-negative on all products of generators and the unit $1$, and additionally needs to satisfy
\beqn
\label{consprobstate}
\rho\left(x(v_a+v_{\overline{a}})\right)=\rho(x),\quad\rho\left(x(v_b+v_{\overline{b}})\right)=\rho(x)
\eeqn
for any $x\in\mathcal{A}_{gp}$. The set of unnormalized states $\Omega_0$ is a convex cone in the vector space dual $\mathcal{A}_{gp}^*$. The trace \added{functional} is defined to be
\beq
\tr(\rho)\equiv\rho(1),
\eeq
so that a state is normalized if and only if $\rho(1)=1$. \added{Thereby the state space $\Omega$ is defined.}

Now for the definition of the operators $a$, $\overline{a}$, $b$, $\overline{b}$\added{, which should map $\Omega_0$ to itself}. Given an unnormalized state $\rho\in\Omega_0$, they produce a new state which is defined as
\beqa
a(\rho)(x)&\equiv&\rho(\replaced{v_ax}{xv_a})\\
\overline{a}(\rho)(x)&\equiv&\rho(\replaced{v_{\overline{a}}x}{xv_{\overline{a}}})\\
b(\rho)(x)&\equiv&\rho(\replaced{v_bx}{xv_b})\\
\overline{b}(\rho)(x)&\equiv&\rho(\replaced{v_{\overline{b}}x}{xv_{\overline{b}}})
\eeqa
Since $v_a^2=v_a$, it follows that $a^2=a$, and \replaced{similarly it follows that $\overline{a}^2=\overline{a}$, $b^2=b$ and $\overline{b}^2=\overline{b}$ hold true.}{analogous derivations work for all the other relations required to hold between the $a$, $\overline{a}$, $b$, $\overline{b}$.}

Now given any initial state $\rho$ and conducting the alternating measurements of $a$ and $b$, the model predicts outcome probabilities that are given by
\beqn
\label{gpprobs}
\begin{split}
P_a(r)=\rho\left(v_{a^{r_1}}v_{b^{r_2}}v_{a^{r_3}}\ldots\right)\\
P_b(r)=\rho\left(v_{b^{r_1}}v_{a^{r_2}}v_{b^{r_3}}\ldots\right)
\end{split}
\eeqn
So given any assignment of outcome probabilities $P_a$, $P_b$, one can regard the equations\re{gpprobs} as a definition of $\rho$ on products of generators. This $\rho$ extends to a state on $\mathcal{A}_gp$ by linearity, where the equations\replaced{\re{consprobstate}}{\re{outcomeprob}} hold by conservation of probability\added{\re{conserved}}. This ends the proof.
\end{proof}

\section{Remarks on potential generalizations}
\label{gens}

It would certainly be desirable to generalize the present results about quantum mechanics to situations involving a higher number of measurements or a higher number of outcomes per measurement or by allowing non-trivial dynamics for the system. I will now describe the corresponding $C^*$-algebras involved in this which one would have to understand in order to achieve such a generalization.

Consider a ``black box'' system analogous to figure~\ref{bb} on which the experimenter can conduct $k$ different kinds of measurement. Suppose also that the $j$th measurement has $n_j\in\N$ \added{possible} outcomes, and that again these measurements are repeatable, which again implies the absence of non-trivial dynamics.

A quantum-mechanical observable describing a von Neumann measurement with $n$ possible outcomes is given by a hermitian operator with (up to) $n$ different eigenvalues. Since the eigenvalues are nothing but arbitrary labels of the measurement outcomes, we might as well label the outcomes by the roots of unity $e^{\frac{2\pi i l}{n}}$, $l\in\{0,\ldots,n-1\}$. But then in this case the observable is given by a unitary operator $u$ which satisfies $u^n=1$. Conversely, given any unitary operator $u$ of order $n$, we can diagonalize $u$ into eigenspaces with eigenvalues being the roots of unity $e^{\frac{2\pi i l}{n}}$, and therefore we can think of $u$ as being an observable where the $n$ outcomes are labelled by the \added{$n$th} roots of unity.

By this reasoning, the specification of $k$ observables where the $j$th observable has $n_j$ different outcomes is equivalent to specifying $k$ unitary operators, where the $j$th operator is of order $n_j$. Hence, the corresponding universal $C^*$-algebra is in this case given by the $C^*$-algebra freely generated by unitaries of the appropriate orders. But this object in turn coincides with the \added{maximal} group $C^*$-algebra
\beq
C^*(\Z_{n_1}\ast\ldots\ast\Z_{n_k})
\eeq
where the group is the indicated free product of finite cyclic groups. One should expect that these $C^*$-algebras have a very intricate structure in general; for example when $k=2$ and $n_1=2$, $n_2=3$, one has the well-known isomorphism $\Z_2\ast\Z_3\cong PSL_2(\Z)$, so that one has to deal with the \added{maximal} group $C^*$-algebra of the modular group.

\section{Possible experimental tests of quantum mechanics}
\label{experiments}

The results of the previous sections show that the quantum region is certainly much smaller in the space of all probabilities than the general probabilistic region. Therefore, specific experimental tests of the quantum constraints from theorem~\ref{mainthm} in a finite truncation seem indeed appropriate. \added{Among the obvious requirements for such an experiment are
\begin{itemize}
\item One needs a system with two dichotomic observables, which are very close to ideal von Neumann measurements.
\item It has to be possible to measure these observables without destroying the observed system.
\end{itemize}}
\deleted{However,} There is an\added{other} important caveat: for sufficiently small systems with many symmetries, it \replaced{can}{might} be the case that any general probabilistic model is automatically a quantum theory. For example, when the convex set of states of a general probabilistic theory lives in $\R^3$ together with its usual action of the rotation group $SO(3)$ as symmetries, then it is automatically implied that the system is described by quantum mechanics, since every bounded and rotationally invariant convex set in $\R^3$ is a ball and therefore affinely isomorphic to the quantum-mechanical Bloch ball. This observation shows that some obvious candidates for experimental tests---like a photon sent through two kinds of polarizers\added{ with different orientations}---are too small for a successful distinction of quantum theory vs. different general probabilistic theories \added{along the lines proposed in this paper. On the other hand, genuinely dichotomic measurements are hard to come by on bigger systems, as this requires a high level of degeneracy}. \deleted{For these reasons, one should try to find systems with many degrees of freedom and few symmetries in order to conceive of experimental tests of quantum mechanics.} The \added{three-photon experiment or the} quantum dot experiment described in section~\ref{prelim} might be \deleted{a} good starting point\added{s} for further investigation of \added{all of} these issues.

\section{Conclusion}
\label{concl}

This \replaced{paper}{article} was concerned with the simplest non-trivial case of the \replaced{representation}{inverse} problem of quantum measurement for iterated measurements: given the probabilities for outcomes of sequences of iterated measurements on a physical systems, under which conditions can there exist a quantum-mechanical model of the system which \replaced{represents}{recovers} these probabilities? This question has been answered \added{by theorem~\ref{mainthm}} to the extent that there are several infinite sequences of constraints, all of which \replaced{come}{are} rather unexpected (at least to the author). They show that the quantum region in the space of all probabilities is actually quite small and comparatively low-dimensional. \added{On the other hand, theorem~\ref{gpthm} shows that every point in the space of all probabilities can be represented by a general probabilistic model}. In this sense, quantum-mechanical \replaced{models are of}{theories belong to} a very specific \replaced{kind}{class}. The present results yield no insight on the question why our world should be quantum-mechanical---to the contrary, the conditions in theorem\re{mainthm} are so unituitive and complicated that the existence of a direct physical reason for their presence seems unlikely.

A clearly positive feature of the strict constraints for quantum-mechanical models is that they could facilitate further experimental tests of quantum mechanics.

\begin{appendix}
\section{Appendix: Two noncommutative moment problems}
\label{app}

Let $\mathcal{A}\equiv\mathscr{C}\left([0,1],M_2(\C)\right)$ be the $C^*$-algebra of continuous functions with values in $2\!\times\!2$-matrices. The variable of these matrix-valued functions is denoted by $t\in[0,1]$.

\begin{theorem}
\label{ncmp1}
Given real numbers $M_1(n,s)$, $M_x(n,s)$ and $M_z(n,s)$ for each $n\in\N_0$ and \mbox{$s\in\{0,\ldots,n\}$}, there exists a state $\rho$ on $\mathcal{A}$ that has the moments
\beqn
\label{rels}
\begin{split}
M_1(n,s)=\rho\left(t^{n-s}(1-t)^s\cdot\mathbbm{1}_2\right)\\
M_x(n,s)=\rho\left(t^{n-s}(1-t)^s\cdot\sigma_x\right)\\
M_z(n,s)=\rho\left(t^{n-s}(1-t)^s\cdot\sigma_z\right)
\end{split}
\eeqn
if and only if the following conditions hold:
\begin{itemize}
\item probability conservation:
\beqn
\label{probconserv}
M_i(n,s)=M_i(n+1,s)+M_i(n+1,s+1)\quad\forall i\in\{1,x,z\}
\eeqn
\item non-negativity:
\beqn
\label{nonneg}
M_1(n,s)\geq\sqrt{M_x(n,s)^2+M_z(n,s)^2}
\eeqn
\item normalization:
\beqn
\label{normalized}
M_1(0,0)=1
\eeqn
\end{itemize}
\end{theorem}

\begin{proof}
This proof is an adaptation of the solution of the Hausdorff moment problem as it is outlined in~\cite[III \textsection 2]{Wi}. Given the state $\rho$, it follows that\re{probconserv} holds by $1=t+(1-t)$. \deleted{The non-negativity of $M_1(n,s)$ is a consequence of the fact that the function $t^{n-s}(1-t)^s\cdot\mathbbm{1}_2$ is a positive element of $\mathcal{A}$.} For the \replaced{non-negativity}{second} inequality, note that the linear combination
\beq
c\,\mathbbm{1}_2+r\,\sigma_x+s\,\sigma_z
\eeq
is a positive matrix if and only if both the determinant and the trace are non-negative, which means that $r^2+s^2\leq c^2$ and $c\geq 0$. Hence in this case, the function
\beq
t^{n-s}(1-t)^s\cdot\left(c\,\mathbbm{1}_2+r\,\sigma_x+s\,\sigma_z\right)
\eeq
is a positive element of $\mathcal{A}$, and the assertion follows by applying $\rho$ to this function and choosing the values
\beq
r=-M_x(n,s),\quad s=-M_z(n,s),\quad c=\sqrt{M_x(n,s)^2+M_z(n,s)^2}.
\eeq

The main burden of the proof is to construct a state $\rho$, given \replaced{moments which satisfy}{a sequence of moments that satisfies} the constraints\re{probconserv},\re{nonneg} \added{and\re{normalized}}. First of all,\re{probconserv} \deleted{, \re{nonneg} and\re{normalized} together} impl\replaced{ies}{y} that
\beqn
\label{below1}
\added{M_i(n,s)=\sum_{r=s}^{k-n+s}\binom{k-n}{r-s}M_i(k,r),\quad \forall k\geq n,\: i\in\{1,x,z\},}\quad\deleted{M_1(n,s)\geq 1\quad\forall n,s}
\eeqn
\added{which can be proven by induction on $k$. Since the binomial coefficient vanishes in that case, it is also possible to sum from $k=0$ up to $r=k$ without changing the left-hand side.}

Now denote by $\mathcal{P}$ the real vector space of $\R[t]$-linear combinations of the matrices $\mathbbm{1}_2$, $\sigma_x$ and $\sigma_z$. The state $\rho$ will first be constructed on $\mathcal{P}$, which is a real linear subspace of $\mathcal{A}$.

Recall that the Bernstein polynomials~\cite{Lo}
\beq
B_{n,s}(t)=\binom{n}{s}t^s(1-t)^{n-s}
\eeq
can be used to approximate any continuous function on $[0,1]$ in the sense that the approximants
\beq
A_n(f)(t)\equiv\sum_{s=0}^nf\left(\frac{s}{n}\right)B_{n,s}(t)
\eeq
converge uniformly to $f$,
\beq
\left|f(t)-A_n(f)(t)\right|<\varepsilon_n\:\:\forall t\in[0,1],\quad\:\varepsilon_n\stackrel{n\ra\infty}{\lra}0.
\eeq
The Bernstein polynomials can be used to construct a sequence of approximating states $\rho_n$ on $\mathcal{P}$\added{, $n\in\N$}. The\deleted{se} \added{$\rho_n$} are defined in terms of the given moments as
\beq
\rho_n\left(P_1(t)\mathbbm{1}_2+P_x(t)\sigma_x+P_z(t)\sigma_z\right)
\eeq
\beq
\equiv\sum_{s=0}^n\binom{n}{s}\left[P_1\left(\frac{s}{n}\right)M_1(n,s)+P_x\left(\frac{s}{n}\right)M_x(n,s)+P_z\left(\frac{s}{n}\right)M_z(n,s)\right].
\eeq
for any polynomials $P_1$, $P_x$ and $P_z$. Although it is hard to directly check convergence of the sequence $\left(\rho_n\right)_{n\in\N}$, it is at least clear that the $\rho_n$ are uniformly bounded,
\beqn
\label{bounded}
\begin{split}
|\rho_n(P_1(t)\mathbbm{1}_2+&P_x(t)\sigma_x+P_z(t)\sigma_z)\,|\\
\leq&\added{\sum_{s=0}^n\binom{n}{s}\Bigg[\left|P_1\left(\frac{s}{n}\right)\right|M_1(n,s)+\sqrt{P_x\left(\frac{s}{n}\right)^2+P_z\left(\frac{s}{n}\right)^2}\cdot}\\
&\added{\cdot\bigg|\frac{P_x(\frac{s}{n})}{\sqrt{P_x\left(\frac{s}{n}\right)^2+P_z\left(\frac{s}{n}\right)^2}}M_x(n,s)+\frac{P_z(\frac{s}{n})}{\sqrt{P_x\left(\frac{s}{n}\right)^2+P_z\left(\frac{s}{n}\right)^2}}M_z(n,s)\bigg|\Bigg]}\\
\added{\stackrel{\re{nonneg}}{\leq}}&\added{\sum_{s=0}^n\binom{n}{s}\Bigg[\left|P_1\left(\frac{s}{n}\right)\right|M_1(n,s)+\sqrt{P_x\left(\frac{s}{n}\right)^2+P_z\left(\frac{s}{n}\right)^2}\cdot M_1(n,s)\Bigg]}\\
\stackrel{\re{below1},\re{normalized}}{\leq}&\max_{t\in[0,1]}\left[\left|P_1(t)\right|+\sqrt{P_x(t)^2+P_z(t)^2}\right]\\
\stackrel{\phantom{\re{nonneg}}}{=}&\max_{t\in[0,1]}\left|\left|P_1(t)\mathbbm{1}_2+P_x(t)\sigma_x+P_z(t)\sigma_z\right|\right|
\end{split}
\eeqn
where the \deleted{norm in the} last expression coincides with the $C^*$-algebra norm on $\mathcal{A}$.

On the other hand, let $\mathcal{P}_n$ be the subspace of $\mathcal{P}$ where the polynomials are of degree up to $n$. A basis of $\mathcal{P}_n$ is given by the $3n+3$ matrix\added{-valued} polynomials
\beqn
\label{basispolys}
B_{n,s}\mathbbm{1}_2,\:\: B_{n,s}\sigma_x,\:\: B_{n,s}\sigma_z;\quad s\in\{0,\ldots,n\}.
\eeqn
Then the requirements\re{rels} uniquely define a linear functional $\widetilde{\rho}_k:\mathcal{P}_k\ra\R$,
\beqa
\widetilde{\rho}_k\left(B_{n,s}\mathbbm{1}_2\right)&=&M_1(n,n-s)\\
\widetilde{\rho}_k\left(B_{n,s}\sigma_x\right)&=&M_x(n,n-s)\\
\widetilde{\rho}_k\left(B_{n,s}\sigma_z\right)&=&M_z(n,n-s).
\eeqa
But now the relations
\beq
\frac{B_{n,s}}{\binom{n}{s}}=\frac{B_{n+1,s}}{\binom{n+1}{s}}+\frac{B_{n+1,s+1}}{\binom{n+1}{s+1}},
\eeq
in conjunction with the additivity law\re{probconserv}, show that the diagram
\beq
\xymatrix{{}\mathcal{P}_k\ar[rr]\ar[rd] && {\mathcal{P}}_{k+1}\ar[ld]\\ & {\mathbb{R}}}
\eeq
commutes for all $k$. Therefore, the $\widetilde{\rho}_k$ extend to a \replaced{linear functional}{trial state} $\widetilde{\rho}:\mathcal{P}\ra\R$, which is now defined on all of $\mathcal{P}$. The problem with $\widetilde{\rho}$ is that \replaced{its boundedness is hard to check}{it is not obviously bounded}.

Therefore, the rest of this proof is devoted to showing that the approximating states converge to the trial state in the weak sense:
\beq
\rho_k(P)\stackrel{k\ra\infty}{\lra}\widetilde{\rho}(P)\quad\forall P\in\mathcal{P}. 
\eeq
Then\re{bounded} implies that $\widetilde{\rho}$ is bounded and $||\widetilde{\rho}||=1$. Hence the Hahn-Banach extension theorem shows that $\widetilde{\rho}$ can be extended to a linear functional $\widehat{\rho}:\mathcal{A}\ra\C$ with $||\widehat{\rho}||=1$. This proves the original assertion by the fact that this is automatically a state as soon as $||\widehat{\rho}\,||=\added{\widehat{\rho}}(\mathbbm{1})=1$ holds, and the construction of $\widehat{\rho}$ such that the equations\re{rels} hold for this state.

In order to check this convergence, it is sufficient to consider the values of the states on the basis polynomials\re{basispolys}. And for those, the calculation will be shown only for the first type $B_{n,s}\mathbbm{1}_2$, since the other two work in exactly the same way.
\begin{align*}
\begin{split}
\widetilde{\rho}\left(B_{n,n-s}(t)\mathbbm{1}_2\right)- & \rho_k\left(B_{n,n-s}(t)\mathbbm{1}_2\right)\\
&=\added{\binom{n}{s}}M_1(n,s)-\added{\binom{n}{s}}\sum_{r=0}^k\binom{k}{r}\left(\frac{r}{k}\right)^{n-s}\left(1-\frac{r}{k}\right)^sM_1(k,r)\\
&\stackrel{(\replaced{\ref{below1}}{\ref{probconserv}})}{=}\added{\binom{n}{s}}\sum_{r=0}^k\left[\binom{k-n}{r-s}-\binom{k}{r}\left(\frac{r}{k}\right)^{n-s}\left(1-\frac{r}{k}\right)^s\right]M_1(k,r)\\
&=\added{\binom{n}{s}}\sum_{r=0}^k\left[\frac{\binom{k-n}{r-s}}{\binom{k}{r}}-\left(\frac{r}{k}\right)^{n-s}\left(1-\frac{r}{k}\right)^s\right]\binom{k}{r}M_1(k,r)
\end{split}
\end{align*}
Therefore using $\sum_{r=0}^k\binom{k}{r}M_1(k,r)=M_1(0,0)=1$,
\beqn
\label{gammastuff}
\left|\widetilde{\rho}\left(B_{n,n-s}(t)\mathbbm{1}_2\right)-\rho_k\left(B_{n,n-s}(t)\mathbbm{1}_2\right)\right|\leq\added{\binom{n}{s}}\max_{r=0}^k\left|\frac{\binom{k-n}{r-s}}{\binom{k}{r}}-\left(\frac{r}{k}\right)^{n-s}\left(1-\frac{r}{k}\right)^s\right|
\eeqn
\beq
\leq\added{\binom{n}{s}}\max_{y\in[0,1]}\left|\frac{\Gamma(k-n+1)}{\Gamma(k+1)}\cdot\frac{\Gamma(ky+1)}{\Gamma(ky-s+1)}\cdot\frac{\Gamma(k(1-y)+1)}{\Gamma(k(1-y)-n+s+1)}-y^{n-s}(1-y)^s\right|
\eeq
This expression trivially vanishes for $y=0$ and for $y=1$. For $y\in(0,1)$, all the Gamma function arguments tend to infinity, therefore the formula
\beq
\lim_{t\ra\infty}\frac{\Gamma(t+m+1)}{\Gamma(t+1)}\cdot t^{-m}=1
\eeq
can be applied in the form
\beq
\left|\frac{\Gamma(t+m+1)}{\Gamma(t+1)}-t^m\right|<\varepsilon\cdot t^m\quad\forall t\geq t_0(m,\varepsilon)
\eeq
to show that\re{gammastuff} vanishes in the $k\ra\infty$ limit. This finally ends the proof.
\end{proof}

Before studying the \replaced{second}{next} noncommutative moment problem, some preparation is needed. So let $\mathcal{C}\subseteq\R^4$ be the set of points $(x_0,y_0,x_1,y_1)\in\R^4$ with the following property: the rectangle in $\R^2$ that is spanned by $(x_0,y_0)$ as the lower left corner and $(x_1,y_1)$ as the upper right corner has non-empty intersection with the unit disc $\{(x,y)\,|\,x^2+y^2\leq 1\}$.

\begin{proposition}
\label{C}
$\mathcal{C}$ is a convex semialgebraic set. A point $(x_0,y_0,x_1,y_1)$ lies in $\mathcal{C}$ if and only if it satisfies all the following five clauses:
\beqa
&x_0\leq x_1\:\land\:y_0\leq y_1\\
&\left(x_0\leq\:1\land y_0\leq 0\right)\lor\left(x_0\leq 0\land y_0\leq\:1\right)\lor\left(x_0^2+y_0^2\leq 1\right)\\
&\left(x_1\geq -1\land y_0\leq 0\right)\lor\left(x_1\geq 0\land y_0\leq\:1\right)\lor\left(x_1^2+y_0^2\leq 1\right)\\
&\left(x_1\geq -1\land y_1\geq 0\right)\lor\left(x_1\geq 0\land y_1\geq -1\right)\lor\left(x_1^2+y_1^2\leq 1\right)\\
&\left(x_0\leq\: 1\land y_1\geq 0\right)\lor\left(x_0\leq 0\land y_1\geq -1\right)\lor\left(x_0^2+y_1^2\leq 1\right)
\eeqa
\end{proposition}

\begin{proof}
$\mathcal{C}$ is the projection obtained by forgetting the first two coordinates of the points in the set
\beq
\widetilde{\mathcal{C}}\equiv\left\{\left(x,y,x_0,y_0,x_1,y_1\right)\in\R^6\:\big|\:x_0\leq x\leq x_1,\:y_0\leq y\leq y_1,\:x^2+y^2\leq 1\right\}.
\eeq
Since $\widetilde{\mathcal{C}}$ is convex semi-algebraic, so is any projection of it, and therefore $\mathcal{C}$.

A description of $\widetilde{\mathcal{C}}$ in terms of linear inequalities is given by
\beqa
&-x+x_0\leq 0,\quad x-x_1\leq 0\\
&-y+y_0\leq 0,\quad y-y_1\leq 0\\
&x\cdot\cos\alpha+y\cdot\sin\alpha\leq 1\quad\forall\alpha\in[0,2\pi]
\eeqa
From this, one obtains the linear inequalities that define $\mathcal{C}$ by taking all these positive linear combinations for which the dummy variables $x$ and $y$ drop out. There are exactly two such combinations that do not use the $\alpha$-family inequalities, and they are $x_0\leq x_1$ and $y_0\leq y_1$. On the other hand, if such a linear combination contains $\alpha$-family inequalities for \deleted{at least} two \added{or more} different values of $\alpha$, the inequality cannot be tight, since any non-trivial positive linear combination of the $\alpha$-family inequalities for different values of $\alpha$ is dominated by a single one with another value of $\alpha$. Therefore, it suffices to conisder each value of $\alpha$ at a time, and add appropriate multiples of the other inequalities such that $x$ and $y$ drop out. Since for both $x$ and $y$ and each sign, there is exactly one inequality among the first four that contains that variable with that sign, there is a unique way to replace $x$ by $x_0$ or $x_1$ and \added{a unique way to replace} $y$ by $y_0$ or $y_1$. Depending on the value of $\alpha$, there are four sign combinations to consider, and the result is the following set of inequalities:
\beqa
&&x_0\cdot\cos\alpha+y_0\cdot\sin\alpha\leq 1\quad\forall\alpha\in[0,\pi/2],\\
&&x_1\cdot\cos\alpha+y_0\cdot\sin\alpha\leq 1\quad\forall\alpha\in[\pi/2,\pi],\\
&&x_1\cdot\cos\alpha+y_1\cdot\sin\alpha\leq 1\quad\forall\alpha\in[\pi,3\pi/2],\\
&&x_0\cdot\cos\alpha+y_1\cdot\sin\alpha\leq 1\quad\forall\alpha\in[3\pi/2,2\pi].
\eeqa
Each of these families of inequalities in turn is equivalent to the corresponding clause above; for example, $\alpha\in[0,\pi/2]$ bounds a region defined by the lines $x_0=1$, $y_0=1$ and the circular arc in the first quadrant of the $x_0$-$y_0$-plane. This region coincides with the one defined by the first of the clauses above. This works in the same way for the other three families.
\end{proof}

\begin{theorem}
\label{ncmp2}
Given real numbers $M'_1(n,s)$, $M'_x(n,s)$ and $M'_z(n,s)$ for each $n\in\N_0$ and \mbox{$s\in\{0,\ldots,n\}$}, there exists a state $\rho$ on $\mathcal{A}$ that has the (integer and half-integer) moments
\beqn
\label{rels2}
\begin{split}
&M'_1(n,s)=\rho\left(t^{n-s}(1-t)^s\cdot\mathbbm{1}_2\right)\\
&M'_x(n,s)=\rho\left(t^{n-s+1/2}(1-t)^s\cdot\sigma_x\right)\\
&M'_z(n,s)=\rho\left(t^{n-s}(1-t)^{s+1/2}\cdot\sigma_z\right)
\end{split}
\eeqn
if and only if \added{all of these numbers lie in $[-1,+1]$ and} the following \added{additional} conditions hold:
\begin{itemize}
\item probability conservation:
\beqn
\label{probconserv2}
M'_i(n,s)=M'_i(n+1,s)+M'_i(n+1,s+1)\quad\forall i\in\{1,x,z\}
\eeqn
\item non-negativity:
\beqn
\label{nonneg2}
M'_1(n,s)\geq 0
\eeqn
for all $n\in\N_0$ and $s\in\{0,\ldots,n\}$. Furthermore,\footnote{Note that all sums are automatically absolutely convergent since $|M_i|\leq 1$ and $\sum_{k=0}^\infty\left|\binom{1/2}{k}\right|=1<\infty$.}
\beqn
\label{newng}
\begin{split}
&\left(\sum_{k=0}^\infty(-1)^k\binom{1/2}{k}M'_x(n+k-1,s+k)\right)^2\\
+&\left(\sum_{k=0}^\infty(-1)^k\binom{1/2}{k}M'_z(n+k-1,s-1)\right)^2\leq M'_1(n,s)^2
\end{split}
\eeqn
for $n\in\N$ and $s\in\{1,\ldots,n-1\}$. Finally, using the coefficients
\beq
c_{n,k}=(-1)^k\binom{-1/2}{k}-(-1)^{k-n}\binom{-1/2}{k-n}
\eeq
and the quantities
\beq
\begin{split}
V_{x,\pm}(n)=&\sum_{k=0}^\infty c_{n,k}M'_x(k,k)\\
&\pm\sqrt{M'_1(n,n)^2-\left(\sum_{k=0}^\infty(-1)^k\binom{1/2}{k}M'_z(n+k-1,n-1)\right)^2}\\
V_{z,\pm}(n)=&\sum_{k=0}^\infty c_{n,k}M'_z(k,0)\\
&\pm\sqrt{M'_1(n,0)^2-\left(\sum_{k=0}^\infty(-1)^k\binom{1/2}{k}M'_x(n+k-1,k)\right)^2}
\end{split}
\eeq
the point in $\R^4$ given by
\beqn
\label{point}
\left(\sup_{n}V_{x,-}(n),\:\sup_nV_{z,-}(n),\:\inf_{n}V_{x,+}(n),\:\inf_nV_{z,+}(n)\right)
\eeqn
has to lie in \added{the convex region} $\mathcal{C}$ \added{characterized in proposition\re{C}}.\footnote{In particular, the expressions under the square roots have to be non-negative and the suprema and infima have to be finite.}
\item normalization:
\beqn
\label{normalized2}
M'_1(0,0)=1
\eeqn
\end{itemize}
\end{theorem}

\begin{proof}
It will be shown first that these conditions are necessary. This is immediate for\re{probconserv2},\re{nonneg2} and\re{normalized2}. Furthermore, the (uniformly convergent) binomial expansions
\beq
\sqrt{t}=\sqrt{1-(1-t)}=\sum_{k=0}^\infty(-1)^k\binom{1/2}{k}(1-t)^k\\
\eeq
\beq
\sqrt{1-t}=\sum_{k=0}^\infty(-1)^k\binom{1/2}{k} t^k
\eeq
can be applied to express most of the integer moments of a given state in terms of the half-integer moments of that state,
\beqn
\label{recover}
\begin{split}
\rho\left(t^{n-s}(1-t)^s\sigma_x\right)&=\sum_{k=0}^\infty(-1)^k\binom{1/2}{k}\rho\left(t^{n-s-1/2}(1-t)^{s+k}\sigma_x\right),\quad s\in\{0,\ldots,n-1\}\\
\rho\left(t^{n-s}(1-t)^s\sigma_z\right)&=\sum_{k=0}^\infty(-1)^k\binom{1/2}{k}\rho\left(t^{n-s+k}(1-t)^{s-1/2}\sigma_z\right),\quad s\in\{1,\ldots,n\}.
\end{split}
\eeqn
In the present notation\re{rels} and\re{rels2}, this reads
\beqn
\begin{split}
\label{integerfromhalf}
M_x(n,s)=\sum_{k=0}^\infty(-1)^k\binom{1/2}{k}M'_x(n+k-1,s+k)&,\quad s\in\{0,\ldots,n-1\}\\
M_z(n,s)=\sum_{k=0}^\infty(-1)^k\binom{1/2}{k}M'_z(n+k-1,s-1)&,\quad s\in\{1,\ldots,n\}.
\end{split}
\eeqn
Together with\re{nonneg}, these formulas imply the constraint\re{newng} for all relevant values \mbox{$s\in\{1,\ldots,n-1\}$}. Given in addition $M_x(0,0)=\rho(\sigma_x)$ and $M_z(0,0)=\rho(\sigma_z)$, the \replaced{missing}{remaining} integer moments undetermined by\re{integerfromhalf} can be calculated as
\beqn
\label{missingmoments}
\begin{split}
&M_x(n,n)\stackrel{(\ref{probconserv})}{=}M_x(0,0)-\sum_{k=1}^nM_x(k,k-1)\stackrel{(\ref{integerfromhalf})}{=}M_x(0,0)-\sum_{k=0}^\infty c_{n,k}M'_x(k,k),\\
&M_z(n,0)\stackrel{(\ref{probconserv})}{=}M_z(0,0)-\sum_{k=1}^nM_z(k,1)\stackrel{(\ref{integerfromhalf})}{=}M_z(0,0)-\sum_{k=0}^\infty c_{n,k}M'_z(k,0).
\end{split}
\eeqn
where the second steps also involve rearrangements of the sums. Since $M_x(n,n)$ is constrained by\re{nonneg} to have an absolute value of at most
\beq
\sqrt{M_1(n,n)^2-M_z(n,n)^2}=\sqrt{M'_1(n,n)^2-\left(\sum_{k=0}^\infty\binom{1/2}{k}M'_z(n+k-1,n-1)\right)^2},
\eeq
equation\re{missingmoments} shows that $M_x(0,0)$ has to lie in the interval
\beqn
\label{interval}
[V_{x,-}(n),V_{x,+}(n)]
\eeqn
for all $n$; therefore, it also has to lie in the intersection of all these intervals, which is the interval
\beq
\left[\sup_nV_{x,-}(n),\inf_nV_{x,+}(n)\right].
\eeq
Exactly analogous considerations show that $M_z(0,0)$ has to lie in the interval
\beq
\left[\sup_nV_{z,-}(n),\inf_nV_{z,+}(n)\right].
\eeq
Now one concludes that the point\re{point} has to be in $\mathcal{C}$ by the additional constraint 
\beqn
\label{disk}
M_x(0,0)^2+M_z(0,0)^2\leq M_1(0,0)^2=1.
\eeqn

For the converse direction, it will be shown that the assumptions imply the existence of moments $M_x(n,s)$ and $M_z(n,s)$ satisfying the hypotheses of theorem~\ref{ncmp1} such that the $M'_x$ and $M'_z$ can be recovered as
\beqn
\label{halffrominteger}
\begin{split}
&M'_x(n,s)=\sum_{k=0}^\infty(-1)^k\binom{1/2}{k}M_x(n+k,s+k)\\
&M'_z(n,s)=\sum_{k=0}^\infty(-1)^k\binom{1/2}{k}M_z(n+k,s),
\end{split}
\eeqn
and such that the $M_1(n,s)$ coincide with the $M'_1(n,s)$. To begin, use\re{integerfromhalf} to define $M_x(n,s)$ for $s\in\{0,\ldots,n-1\}$ and $M_z(n,s)$ for $s\in\{1,\ldots,n\}$. As soon as additionally the values for $M_x(0,0)$ and $M_z(0,0)$ are determined, the remaining integer moments are defined by\re{missingmoments}. Then it can be verified by direct calculation---treating the cases $s\in\{1,\ldots,n-1\}$ separately from $s=0$ and $s=n$---that the equations\re{halffrominteger} hold, independently of the chosen values for $M_x(0,0)$ and $M_z(0,0)$.

It remains to verify that, with these definitions of $M_x$ and $M_z$, the requirements of theorem\re{ncmp1} can be satisfied for appropriate choices of $M_x(0,0)$ and $M_z(0,0)$. The equations\re{probconserv} easily follow by direct calculation, using\re{probconserv2}. Again by the binomial expansions, the second part of\re{nonneg} is directly equivalent to\re{newng} for $s\in\{1,\ldots,n-1\}$. In the case that $s=n>0$, it holds as long as $M_x(0,0)$ is chosen to lie in the interval\re{interval}; a similar statement holds for $s=0$ and $n>0$. For $s=n=0$, the constraint is equivalent to\re{disk} and means that $\left(M_x(0,0),M_z(0,0)\right)$ has to lie in the unit disk of $\R^2$. By the assumption that\re{point} lies in $\mathcal{C}$, it follows that a consistent choice for $M_x(0,0)$ and $M_z(0,0)$ that satisfies all these requirements is indeed possible.
\end{proof}

\end{appendix}
\end{document}